\documentclass{article}
\usepackage[utf8x]{inputenc}
\usepackage{amssymb, amsmath,amsthm,amsfonts}
\usepackage{enumerate}
\usepackage{color}
\usepackage{tikz}
\usetikzlibrary{mindmap}
\usetikzlibrary{fit}
\usetikzlibrary{shapes}
\usepackage{url}
\usepackage{subfig}
\usepackage{todonotes}
\usepackage{authblk}
\usepackage[T1]{fontenc}

\newcommand*\samethanks[1][\value{footnote}]{\footnotemark[#1]}

\title{The Price of Connectivity for Vertex Cover}
\author[1]{Eglantine Camby\thanks{\{ecamby, jcardin, sfiorini\}@ulb.ac.be}}
\author[2]{Jean Cardinal\samethanks}
\author[1]{Samuel Fiorini\samethanks}
\author[3]{\\ Oliver Schaudt\thanks{schaudt@math.jussieu.fr} \thanks{Parts of this research have been carried out during the visit of Oliver Schaudt to Universit\'e Libre de Bruxelles.}}

\affil[1]{\small Universit\'e Libre de Bruxelles (ULB), D\'epartement de Math\'ematique, CP~216, B-1050 Brussels, Belgium}
\affil[2]{\small Universit\'e Libre de Bruxelles (ULB), D\'epartement d'Informatique, CP~212,

B-1050 Brussels, Belgium}
\affil[3]{\small Universit\'e Pierre et Marie Curie, Combinatoire et Optimisation, 4 place Jussieu, 75252 Paris, France}
\date{}
\begin{document}

\newtheorem{observation}{Observation}
\newtheorem{theorem}{Theorem}
\newtheorem{lemma}{Lemma}
\newtheorem{corollary}{Corollary}
\newcommand{\PoC}{PoC}
\newcommand{\VC}{C}
\newcommand{\CVC}{C_c}
\newcommand{\CVCone}{C_{c,1}}
\newcommand{\CVCtwo}{C_{c,2}}
\newcommand{\poc}{\mathop{\textrm{PoC}}}

\maketitle

\begin{abstract}

The vertex cover number of a graph is the minimum number of vertices that are needed to cover all edges. When those vertices are further required to induce a connected subgraph, the corresponding number is called the connected vertex cover number, and is always greater or equal to the vertex cover number.

Connected vertex covers are found in many applications, and the relationship between those two graph invariants is therefore a natural question to investigate. For that purpose, we introduce the {\em Price of Connectivity}, defined as the ratio between the two vertex cover numbers. We prove that the price of connectivity is at most 2 for arbitrary graphs. We further consider graph classes in which the price of connectivity of every induced subgraph is bounded by some real number $t$. We obtain forbidden induced subgraph characterizations for every real value $t \leq 3/2$.

We also investigate critical graphs for this property, namely, graphs whose price of connectivity is strictly greater than that of any proper induced subgraph. Those are the only graphs that can appear in a forbidden subgraph characterization for the hereditary property of having a price of connectivity at most $t$. In particular, we completely characterize the critical graphs that are also chordal.

Finally, we also consider the question of computing the price of connectivity of a given graph. Unsurprisingly, the decision version of this question is NP-hard. In fact, we show that it is even complete for the class $\Theta_2^P = P^{NP[\log]}$, the class of decision problems that can be solved in polynomial time, provided we can make $O(\log n)$ queries to an NP-oracle. This paves the way for a thorough investigation of the complexity of problems involving ratios of graph invariants.    
\end{abstract}


\section{Introduction}

A \textit{vertex cover} of a graph $G$ is a vertex subset $\VC$ such that every edge of $G$ has at least one endpoint in $\VC$. 
The size of a minimum vertex cover of $G$, denoted by $\tau(G)$, is called the \textit{vertex cover number} of $G$. 
The problem of finding a minimum vertex cover in a graph is one of the 21 NP-hard problems identified by Karp in 1972, and has since been intensively studied in the literature.

A well-known variant of the notion of vertex cover is that of {\em connected vertex cover}, defined as a vertex cover $\CVC$ such that the induced subgraph $G[\CVC]$ is connected. (If $G$ is not connected we ask that $G[\CVC]$ has the same number of component as $G$.) The minimum size of such a set, denoted by $\tau_c(G)$, is the \textit{connected vertex cover number} of $G$. A connected vertex cover of size $\tau_c(G)$ is called a \textit{minimum connected vertex cover}.

Our contribution is to study the interdependence of $\tau$ and $\tau_c$, both from a complexity-theoretic point of view and in some hereditary classes of graphs.

Let us first note that every vertex cover $\VC$ of a connected graph $G$ such that $G[\VC]$ has $c$ connected components can be turned into a connected vertex cover of $G$ by adding at most $c-1$ vertices. This directly yields the following observation. 

\begin{observation}\label{obs basic inequality}
For every graph $G$ it holds that $\tau_c (G) \leqslant 2 \tau (G) - 1$.
\end{observation}

As an immediate consequence of Observation \ref{obs basic inequality}, the following holds for every graph $G$ (with at least one edge):
\begin{equation} \label{eqn general poc}
1 \leqslant \tau_c (G) / \tau (G) < 2.
\end{equation}
We define the {\em Price of Connectivity} (\PoC) of a graph $G$ as the ratio $\tau_c(G) / \tau(G)$. Hence we just showed that the Price of Connectivity of any graph lies in the interval $[1,2)$.
We denote by $P_k$ the path on $k$ vertices and by $C_k$ the cycle on $k$ vertices. Note that the upper bound in (\ref{eqn general poc}) is asymptotically sharp in the class of paths and in the class of cycles, in the sense that
$$\lim\limits_{k \rightarrow \infty} \tau_c (P_k) / \tau (P_k) = 2 = \lim\limits_{k \rightarrow \infty} \tau_c (C_k) / \tau (C_k).$$

Our contribution is split into two parts. In the first part, we consider the computational complexity of the problem of deciding whether the \PoC{} of a graph given as input is bounded by some constant $t$. We show the completeness of this problem with respect to a well-defined complexity class in the polynomial hierarchy. In the second part, we investigate graph classes in which the \PoC{} {\em of every induced subgraph} is bounded by a constant $t$ with $t \in [1,2)$. Those classes will be defined by forbidden induced subgraphs. The forbidden subgraph characterizations directly yields polynomial-time algorithms for recognizing graphs in those classes.

We use the following standard notation. If $G$ and $H$ are two graphs we say that $G$ {\em contains} $H$ if $G$ has an induced subgraph isomorphic to $H$. We say that $G$ is \textit{$H$-free} if $G$ does not contain $H$. Furthermore, we say that $G$ is $(H_1,\ldots,H_\ell)$-free if $G$ does not contain $H_i$ for any $i \in \{1,\ldots,\ell\}$.


The Price of Connectivity (as defined here) has been introduced by Cardinal and Levy~\cite{cardinal-levy,levy}, who showed that it was bounded by $2/(1+\varepsilon)$ in graphs with average degree $\varepsilon n$, where $n$ denotes the number of vertices. Other ratios were previously studied. In a companion paper to the present paper, Camby and Schaudt~\cite{cambyschaudt} consider the Price of Connectivity for dominating set. 
Recently, Schaudt~\cite{schaudt} studied the ratio between the connected domination number and the total domination number. Fulman~\cite{fulman} and Zverovich~\cite{zverovich} investigated the ratio between the independence number and the upper domination number. 

\section{Our results}

All the proofs can be found in the next section.

\subsection{Computational Complexity}

The class $\Theta_2^P = P^{\mbox{{\scriptsize NP}}[\log]}$ is defined as the class of decision problems solvable in polynomial time by a deterministic Turing machine 
that is allowed use $\mathcal{O}(\log n)$ many queries to an NP-oracle, where $n$ is the size of the input.

\begin{theorem} 
\label{thm:T2C}
Let $1 < r < 2$ be a fixed rational number.
Given a connected graph $G$, the problem of deciding whether $\tau_c(G)/\tau(G) \leq r$ is $\Theta_2^p$-complete.
\end{theorem}

It is easy to see that the above decision problem belongs to $\Theta_2^p$, since both $\tau$ and $\tau_c$ can be computed using logarithmically many queries to an \textsc{NP}-oracle by binary search.
Thus, Theorem~\ref{thm:T2C} is a negative result: loosely speaking, it tells us that deciding whether the \PoC{} is bounded by some constant is as hard as computing both $\tau$ and $\tau_c$ explicitely.
And this remains true even if the constant is not part of the input.

Our reduction is from the decision problem whether for two given graphs $G$ and $H$ it holds that $\tau(G) \ge \tau(H)$, which is known to be $\Theta_2^p$-complete due to Spakowski and Vogel~\cite{SV00}.
It uses a gadgetry that allows us to compare $\tau$ and $\tau_c$ on a single graph.

\subsection{\PoC-Perfect Graphs}

As Theorem~\ref{thm:T2C} shows, the class of graphs where $\tau_c(G)/\tau(G) \leq r$ holds (for any fixed rational $r \in (1,2)$) is $\Theta_2^p$-complete to recognize.
However, if we restrict our attention to hereditary graph classes, we are able to derive the following results.
Note that our characterizations yield polynomial time recognition algorithms, since the list of forbidden induced subgraphs is finite in each case.

We first consider the hereditary class of graphs $G$ for which $\tau_c(G) = \tau(G)$, referred to as {\em \PoC-Perfect graphs}. 
A similar result had been found by Zverovich~\cite{Connected-Dominant} for dominating set. 
There, the corresponding class is that of $(P_5,C_5)$-free graphs.
\begin{theorem} \label{thm PoC1}
The following assertions are equivalent for every graph $G$~:
\begin{enumerate}[(i)]
 \item For every induced subgraph $H$ of $G$ it holds that $\tau_c (H) = \tau (H)$.
 \item $G$ is $(P_5,C_5,C_4)$-free.
 \item $G$ is chordal and $P_5$-free.
\end{enumerate}
\end{theorem}

The above characterization tells us that the class of \PoC-Perfect graphs properly contains two well-known classes of graphs: split graphs and trivially perfect graphs (see \cite{BLS} for further reference on these classes).
Moreover, it gives rise to the following definition.

\subsection{\PoC-Near-Perfect Graphs}

%
Let $t \in [1,2)$. A graph $G$ is said to be {\em PoC-Near-Perfect} with threshold $t$ if every induced subgraph $H$ of $G$ satisfies $\tau_c(H) \leqslant t \cdot \tau(H)$. 
This defines a hereditary class of graphs for every choice of $t$. 
Theorem \ref{thm PoC1} gives a forbidden induced subgraphs characterization of this class for $t = 1$. 
Our second result gives such a characterization for $t = 4/3$. 

Note that $\tau_c (C_5) / \tau (C_5) = 4/3$ and $\tau_c (P_5) / \tau (P_5) = \tau_c (C_4) / \tau (C_4) = 3/2$. Hence any graph class that does not forbid either $C_5$ or
$P_5$ contains a graph $G$ such that $\tau_c (G) / \tau (G) = 4/3$. Therefore, the characterization of Theorem~\ref{thm PoC1} also holds for the class of graphs $G$ such that
every induced subgraph $H$ satisfies $\tau_c (H) \leqslant t\cdot \tau (H)$, for any $t\in [1, 4/3)$. We now turn our attention to $t = 4/3$, which is the next interesting threshold after $t = 1$.

\begin{theorem}\label{thm PoC 4/3}
The following assertions are equivalent for every graph $G$~:
\begin{enumerate}[(i)]
 \item For every induced subgraph $H$ of $G$ it holds that $\tau_c (H) \leqslant \frac{4}{3} \cdot \tau (H)$.
 \item $G$ is $(P_5,C_4)$-free.
\end{enumerate}
\end{theorem}
By Theorem~\ref{thm PoC 4/3}, $t = 3/2$ is the next interesting threshold after $t = 4/3$. Our third results states that the list of forbidden induced subgraphs for threshold $t = 3/2$ is $(C_6, P_7, \Delta_1, \Delta_2)$, where $\Delta_1$ is the 1-join of two $C_4$'s, and $\Delta_2$ is obtained from $\Delta_1$ by removing any edge incident to the vertex of degree 4 (see Fig.~\ref{Deltas}). 

\begin{theorem}\label{thm PoC 3/2}
The following assertions are equivalent for every graph $G$~:
\begin{enumerate}[(i)]
\item For every induced subgraph $H$ of $G$ it holds that $\tau_c (H) \leqslant \frac{3}{2} \cdot \tau (H)$.
\item $G$ is $(P_7,C_6, \Delta_1, \Delta_2)$-free.
\end{enumerate}
\end{theorem}

Since a chordal and $P_7$-free graph is $(C_6, P_7, \Delta_1, \Delta_2)$-free, we deduce the following corollary from Theorem~\ref{thm PoC 3/2}.

\begin{corollary}
If $G$ is a chordal, $P_7$-free graph then for every induced subgraph $H$ of $G$, it holds that $\tau_c(H) \leqslant 3/2 \cdot \tau(H).$
\end{corollary}

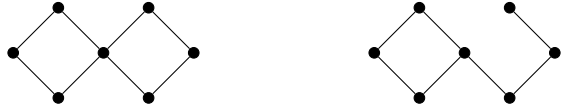
\begin{figure}[!h]
\begin{center}
\begin{tikzpicture}[scale = 0.6]
\tikzstyle{every node}=[draw,circle,fill=black,minimum size=4pt,inner sep=0pt]
\node(a) at (-6,0){};
\node(b) at (-5,1){};
\node(c) at (-5,-1){};
\node(d) at (-4,0){};
\node(e) at (-3,-1){};
\node(f) at (-2,0){};
\node(g) at (-3,1){};
\node(a1) at (6,0){};
\node(b1) at (5,1){};
\node(c1) at (5,-1){};
\node(d1) at (4,0){};
\node(e1) at (3,-1){};
\node(f1) at (2,0){};
\node(g1) at (3,1){};

\draw (d) -- (b) -- (a) -- (c) -- (d) -- (e) -- (f) -- (g) -- (d);
\draw (b1) -- (a1) -- (c1) -- (d1) -- (e1) -- (f1) -- (g1) -- (d1);
\end{tikzpicture}
\caption{An illustration of graphs $\Delta_1$ (on the left) and $\Delta_2$ (on the right).}
\label{Deltas}
\end{center}
\end{figure}

\subsection{\PoC-Critical Graphs}

We now turn our attention to {\em critical graphs}, that is, graphs $G$ for which the \PoC{} of any proper induced subgraph $H$ of $G$ is strictly smaller than the \PoC{} of $G$. 
These are exactly the graphs that can appear in a forbidden induced subgraphs characterization of the PoC-near-perfect graphs for some threshold $t \in [1,2)$. 
A perhaps more tractable class of graphs are the {\em strongly critical} graphs, defined as the graphs $G$ for which every proper (not necessarily induced) subgraph $H$ of $G$ has a \PoC{} that is strictly smaller than the \PoC{} of $G$. 
It is clear that every strongly critical graph is critical, but the converse is not true. 
For instance, $C_5$ is critical, but not strongly critical.

\subsubsection{\PoC-Critical Chordal Graphs}
Let $T$ be a tree.
We call $T$ \textit{special} if it is obtained from another tree by subdividing each edge exactly once and then attaching a pendant vertex to every leaf of the resulting graph (see Fig.~\ref{specialtree} for an example). 
%
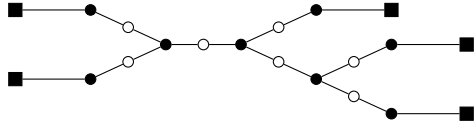
\begin{figure}[!h]
\centering
\begin{tikzpicture}[yscale = 0.23]
\tikzstyle{every node}=[draw,circle,fill=black,minimum size=4pt,inner sep=0pt]
\node (ab) at (0,4) [shape= circle, draw, fill = white] {};
\node (ae) at (-1,3) [shape= circle, draw, fill = white] {};
\node (af) at (-1,5) [shape= circle, draw, fill = white] {};
\node (bc) at (1,5) [shape= circle, draw, fill = white] {};
\node (bd) at (1,3) [shape= circle, draw, fill = white] {};
\node (dg) at (2,3) [shape= circle, draw, fill = white] {};
\node (dh) at (2,1) [shape= circle, draw, fill = white] {};
\node(a) at (-0.5,4){};
\node(b) at (0.5,4){};
\node(c) at (1.5,6){};
\node(d) at (1.5,2){};
\node(g) at (2.5,4){};
\node(h) at (2.5,0){};
\node(f) at (-1.5,6){};
\node(e) at (-1.5,2){};
\node (f1) at (-2.5,6) [rectangle, minimum height=0.5em,minimum
width=0.5em](f1){};
\node (e1) at (-2.5,2) [rectangle, minimum height=0.5em,minimum
width=0.5em](e1){};
\node (c1) at (2.5,6) [rectangle, minimum height=0.5em,minimum
width=0.5em](c1){};
\node (g1) at (3.5,4) [rectangle, minimum height=0.5em,minimum
width=0.5em](g1){};
\node (h1) at (3.5,0) [rectangle, minimum height=0.5em,minimum
width=0.5em](h1){};
\draw (f) -- (af) -- (a) -- (ae) -- (e); 
\draw (a) -- (ab) -- (b) -- (bc) -- (c); 
\draw (b) -- (bd) -- (d) -- (dg) -- (g);
\draw (d) -- (dh) -- (h); 
\draw (f) -- (f1);
\draw (e) -- (e1);
\draw (c) -- (c1);
\draw (g) -- (g1);
\draw (h) -- (h1);
\end{tikzpicture}
\caption{A special tree constructed from another tree (vertices indicated by filled circles) by sudividing each edge exactly once (subdivision vertices are indicated by hollow circles) and by attaching a pendant vertex (indicated by squares) to every leaf of the resulting graph.}
\label{specialtree}
\end{figure}

Our next result characterizes the class of (strongly) critical chordal graphs.
\begin{theorem} \label{thm chordal critical}
For a chordal graph $G$, the following assertions are equivalent~:
\begin{enumerate}[(i)]
 \item $G$ is a special tree.
 \item $G$ is strongly critical.
 \item $G$ is critical.
\end{enumerate}
\end{theorem}

\subsubsection{\PoC-Strongly-Critical Graphs}
Our final  result yields structural constraints on the class of strongly critical graphs. 
\begin{theorem} \label{thm strongly critical structure}
Let $G$ be a strongly critical graph.
\begin{enumerate}[(i)]
 \item Every minimum vertex cover of $G$ is independent. In particular, $G$ is bipartite.
 \item If $G$ has a cutvertex, then $G$ is a special tree.
\end{enumerate}
\end{theorem} 

\section{Proofs}

\subsection{Complexity result}

We now proceed to prove Theorem~\ref{thm:T2C}. 

\begin{lemma}
\label{lem:fixtauc}
Given a connected graph $G$ with $n$ vertices, one can
construct in linear time a graph $G'$ 
such that $\tau(G')=n+\tau(G)$ and $\tau_c(G')=2n$.
\end{lemma}
\begin{proof}
With each vertex $v\in V(G)$, associate three vertices $v,v',v''$ in $V(G')$, and let $E(G'):=E(G)\cup \bigcup_{v\in V(G)} \{ vv', v'v'' \}$.
A minimum vertex cover of $G'$ is the union of a minimum vertex cover of $G$ with all vertices of the form $v'$. On the other hand, 
a minimum connected vertex cover of $G'$ contains all vertices $v, v'$. 
\end{proof}

\begin{figure}
\begin{center}
\includegraphics[width=.3\textwidth]{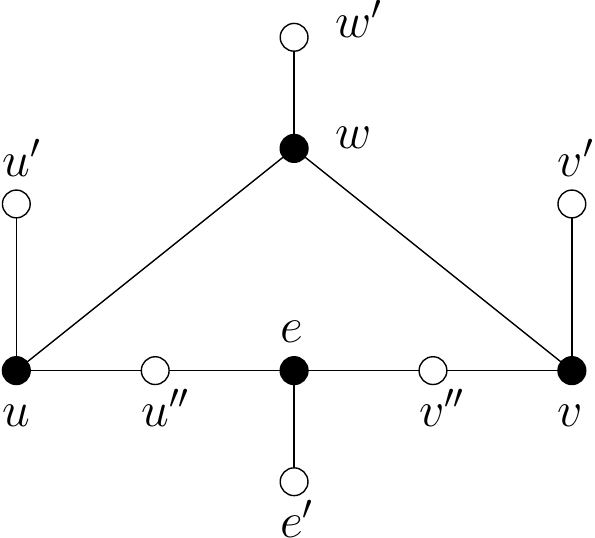}
\end{center}
\caption{\label{fig:edge}Representation of an edge $e=uv$ in the construction of $G'$ in Lemma~\ref{lem:fixtau}.}
\end{figure}

\begin{lemma}
\label{lem:fixtau}
Given a graph $G$ with $n$ vertices and $m$ edges, one can construct in linear time a graph $G'$ such that $\tau(G')=n+m+1$ and $\tau_c(G')=n+m+1+\tau(G)$.
\end{lemma}
\begin{proof}
For each edge $e=uv\in E(G)$, define two vertices $e, e'$ of $V(G')$. For each vertex $v\in V(G)$, define three vertices
$v, v',v''$ of $V(G')$. Finally, add two vertices $w,w'$ to $V(G')$. The set of edges $E(G')$ is defined as follows. For each edge $e=uv\in E(G)$,
the vertices $e$ and $e'$ of $V(G')$ are adjacent, and vertex $e$ is adjacent to vertices $u''$ and $v''$. Similarly, for each vertex $v\in V(G)$, vertices $v$ and $v'$ of $V(G')$ are adjacent, and $v$ is adjacent to both $v''$ and $w$. Finally, $ww'\in E(G')$. The construction is illustrated in Figure~\ref{fig:edge}.

Since for each edge $e\in E(G)$, the corresponding vertex $e\in V(G')$ is adjacent to the degree-one vertex $e'$, it can be considered, without loss of generality, to be part of any minimum vertex cover of $G'$. 
The same remark holds for vertices $v\in V(G)$, and for the unique vertex $w$. 
Now the union $C\subset V(G')$ of those vertices is a vertex cover of $G'$, hence we have $\tau(G')=n+m+1$.

We now have to compute $\tau_c(G')$. The previous vertex cover $C$ is not connected, as $G'[C]$ has exactly $m+1$ connected components: one for each edge of $G$, and one induced by $w$ and the vertices $v\in V(G)$. 
To make it connected, we need to augment $C$ with the fewest possible additional vertices of the form $v''$ for $v\in V(G)$. 
Every such vertex $v''$ will link the component containing $v$ to every vertex $e\in E(G)$ of $G'$ such that $v\in e$. 
Hence the minimum number of additional vertices to add to $C$ is exactly the size $\tau(G)$ of a minimum vertex cover of $G$. 
Thus $\tau_c(G')=n+m+1+\tau(G)$, as claimed.     
\end{proof}

\begin{proof}[of Theorem~\ref{thm:T2C}]
Let $r = r_1 / r_2$ be a fixed rational number with $1<r<2$.
It is clear that the problem is in $\Theta_2^p$, so we proceed to the $\Theta_2^p$-hardness.
Let $G$ and $H$ be two graphs.
We reduce from the $\Theta_2^p$-complete decision problem, whether $\tau(G) \ge \tau(H)$ (see Spakowski and Vogel~\cite{SV00}).

We can assume that $G$ and $H$ are both connected.
Otherwise, we choose a vertex from each connected component of $G$ (resp.~$H$), add two new vertices $w$ and $w'$, and put an edge from $w$ to all chosen vertices and to $w'$.
Let $G'$ (resp.~$H'$) be the graph obtained from $G$ (resp.~$H$) by this procedure.
It is clear that $\tau(G') = \tau(G)+1$ and $\tau(H') = \tau(H)+1$.
Hence, $\tau(G) \ge \tau(H)$ if and only if $\tau(G') \ge \tau(H')$.
So we may assume that both $G$ and $H$ are connected.

The reduction consists of the following five steps.

\textbf{Step 1.}
Let $v$ be any vertex of $G$. 
Starting with $r_2$ disjoint copies of $G$, we connect all $r_2$ copies of $v$ to a new vertex $w$. 
We then attach a pendant vertex $w'$ to $w$. 
The graph obtained we denote by $G_{r_2}$.
Let $n_G=|V(G)|$.
Clearly, $\tau(G_{r_2})= r_2 \tau(G) + 1$ and $|V(G_{r_2})| = r_2 n_G + 2$.

Similarly we construct $H_{r_1}$ from $H$.
Let $n_H = |V(H)|$ and $m_H = |E(H)|$.
Clearly, $\tau(H_{r_1})= r_1 \tau(H) + 1$, $|V(H_{r_1})| = r_1 n_H + 2$, and $|E(H_{r_1})| = r_1 m_H + r_1 + 1$.

\textbf{Step 2.} 
We apply Lemma~\ref{lem:fixtauc} to $G_{r_2}$ to get $G_{r_2}'$.
We obtain
\begin{eqnarray*}
\tau(G_{r_2}') & = & |V(G_{r_2})| + \tau(G_{r_2}) \\
& = & r_2 \tau(G) + r_2 n_G + 3,\\
\tau_c(G_{r_2}') & = & 2 |V(G_{r_2})| \\
& = & 2 r_2 n_G + 4.
\end{eqnarray*}

We apply Lemma~\ref{lem:fixtau} to $H_{r_1}$ to get $H_{r_1}'$, and obtain
\begin{eqnarray*}
\tau (H_{r_1}') & = & |V(H_{r_1})| + |E(H_{r_1})| + 1 \\
& = & r_1 (n_H + m_H + 1) + 4,\\
\tau_c(H_{r_1}') & = & \tau(H_{r_1}) + |V(H_{r_1})| + |E(H_{r_1})| + 1 \\
& = & r_1 \tau(H) + r_1 (n_H + m_H + 1) + 5.
\end{eqnarray*}

\textbf{Step 3.} 
We construct a new graph $U$ by taking the disjoint union of $G_{r_2}'$ and $H_{r_1}'$, and adding an edge $uv$ such that $u\in V(G_{r_2}')$, $v\in V(H_{r_1}')$, 
and both $u$ and $v$ are adjacent to a degree-one vertex in $G_{r_2}'$ and $H_{r_1}'$, respectively (such an edge always exists). 

By construction of $U$,
\begin{eqnarray*}
\tau_c (U) & = & \tau_c(G_{r_2}') + \tau_c(H_{r_1}') \\
& = & r_1 \tau(H) + r_1 (n_H + m_H + 1) + 2 r_2 n_G + 9,\\
\tau (U) & = & \tau(G_{r_2}') + \tau(H_{r_1}') \\
& = & r_2 \tau(G) + r_1 (n_H + m_H + 1) + r_2 n_G + 7 .
\end{eqnarray*}

\textbf{Step 4.}
Let $\varphi_1 = 2 r_2 n_G + r_1 (n_H + m_H + 1) + 9$ and $\varphi_2 = r_2 n_G + r_1 (n_H + m_H + 1) + 7$.
In this step, we determine two non-negative integers $a$ and $b$ such that
\begin{equation}\label{eqn:PoC-a,b}
\frac{a+2b + \varphi_1}{a+b + \varphi_2} = r.
\end{equation}
We claim that the computation of the integers $a$ and $b$ can be done in polynomial time in $\varphi_1 + \varphi_2$ (that is, in the size of $U$) and, moreover, we can choose $a$ and $b$ such that $a,b  \in \mathcal{O}(\varphi_1 + \varphi_2)$.

To see this, consider the affine cone $C \subseteq \mathbb{R}^2$ defined by 
$$C = \{(\varphi_2,\varphi_1) + a (1,1) + b (1,2) : a , b \in \mathbb{R}_{\ge 0}\}$$
and the linear space $L$ defined by
$$L = \{\lambda (1,r) : \lambda \in \mathbb{R}_{\ge 0}\}.$$
Note that
$$\{(\varphi_2,\varphi_1) + a (1,1) + b (1,2) : a,b \in \mathbb{Z}_{\ge 0}\} = C \cap \mathbb{Z}^2.$$
Thus, to find $a$ and $b$, we have to compute an integral point in $C \cap L$.

Since $1 < r < 2$, there is a $\lambda_0$ such that $\lambda (1,r) \in C$ for all $\lambda \ge \lambda_0$.
We claim that we can choose $\lambda_0 \in \mathcal{O}(\varphi_1 + \varphi_2)$.
For such a $\lambda_0$ it holds that there is an integral point $(x,y) \in \{\lambda (1,r) : \lambda_0 \le \lambda \le \lambda_0 + r_1\} \subset C$ and cleary this point can be found in polynomial time in $\varphi_1+\varphi_2$.
Moreover, for the corresponding $a,b$ it holds that $a,b \in \mathcal{O}(\varphi_1+\varphi_2)$.

To see that we can choose $\lambda_0 \in \mathcal{O}(\varphi_1+\varphi_2)$, 
consider the two hyperplanes $H_1 = \{(\varphi_2,\varphi_1) + x (1,1) : x \in \mathbb{R}\}$ and $H_2 = \{(\varphi_2,\varphi_1) + x (1,2) : x \in \mathbb{R}\}.$
It is clear that $C \cap (H_1 \cup H_2)$ is the boundary of $C$.

Let $(x_1,y_1)$ be the unique point in $L \cap H_1$ and let $(x_2,y_2)$ be the unique point in $L \cap H_2$.
Since $C \cap (H_1 \cup H_2)$ is the boundary of the affine cone $C$, we can choose $\lambda_0 = \max\{x_1 , x_2\}$.
A straightforward computation shows that $x_1 = (\varphi_1-\varphi_2) /(r-1)$ and $x_2 = (2 \varphi_2 - \varphi_1) / (2-r)$.
This proves our claim.

\textbf{Step 5.}
We now construct a graph $U'$ from $U$ as follows.
Let $v$ be a vertex in $U$ of degree 1 (such a vertex is always present).
Let $P^1$ be the graph obtained from the chordless path with vertex set $\{u_1, u_2 , \ldots , u_a\}$ by attaching a pendant vertex to every member of $\{u_1, u_2 , \ldots , u_a\}$.
Similarly, let $P^2$ be the graph obtained from the chordless path with vertex set $\{v_1, v_2 , \ldots , v_{2b}\}$ by attaching a pendant vertex to every member of $\{v_2 , v_4 , \ldots , v_{2b}\}$. 
Let $U'$ be the graph obtained from the disjoint union of $U$, $P^1$, and $P^2$ by putting an edge from $v$ to $u_1$ and to $v_1$.
Since $a,b  \in \mathcal{O}(\varphi_1 + \varphi_2)$, the above procedure can be done in linear time in the size of the graph $U$.

By the construction of $U'$, we obtain
\begin{eqnarray*}
\tau_c (U') & = & \tau_c (U) + a + 2b \\
& = & r_1 \tau(H) + a+2b + \varphi_1,\\
\tau (U') & = & \tau(U) + a + b \\
& = & r_2 \tau(G) + a+b + \varphi_2.
\end{eqnarray*}
Recall that $r = r_1 / r_2$. 
By (\ref{eqn:PoC-a,b}), there is some non-negative integer $c$ such that $a+2b + \varphi_1 = r_1 c$ and $a+b + \varphi_2 = r_2 c$.
Hence,
\[
\frac{\tau_c(U')}{\tau(U')} = \frac{r_1 \tau(H) + a+2b + \varphi_1}{r_2 \tau(G) + a+b + \varphi_2}
= \frac{r_1 \tau(H) + r_1 c}{r_2 \tau(G) + r_2 c} = r \frac{\tau(H) + c}{\tau(G) + c}.
\]
Thus, $\tau_c(U')/\tau(U') \le r$ if and only if $\tau(H) \le \tau(G)$.
This completes the proof. 
\end{proof}

\subsection{Structural results}

\begin{lemma} \label{dist2}
Let $G$ be a connected graph and let $\VC$ be a vertex cover of $G$. 
If $(\mathcal{A}, \mathcal{B})$ is a bipartition of the connected components of $\VC$ with $\mathcal{A}, \mathcal{B} \neq \emptyset$, 
there exists $A \in \mathcal{A}$ and $B \in \mathcal{B}$ such that the distance between $A$ and $B$ is exactly $2$.
\end{lemma}
\begin{proof}
Let $(\mathcal{A},\mathcal{B})$ be a bipartition of the connected components of $\VC$. 
Since $\VC$ has a finite number of connected components, there exist $A \in \mathcal{A}$ and $B \in \mathcal{B}$ such that the distance between them is minimum. 
Now we show that this distance is $2$. 
Otherwise, let $x_1x_2 \dots x_n$ be a shortest path between $A$ and $B$ with $x_1 \in A$ and $x_n \in B$, where $n \geqslant 4$. 
In this case, no $x_i$, $i = 2, \dots, n-1,$ belongs to $\VC$. 
Otherwise, $B$ is not one nearest component of $\mathcal{B}$ from $A$ or $A$ is not one nearest component of $\mathcal{A}$ from $B$. 
Thus, the edge $x_2x_3$ is not covered by $\VC$, in contradiction with the definition of vertex cover. 
\end{proof}

\subsection{PoC-Perfect graphs}

\begin{proof}[Theorem~\ref{thm PoC1}]
The class of graphs that are chordal and do not contain an induced $P_5$ is exactly the class of $(C_4,C_5,P_5)$-free graphs.
Since $\tau_c (C_4)/\tau(C_4)=\tau_c(P_5)/\tau(P_5) = 3/2$, and $\tau_c (C_5)/\tau(C_5)=4/3$, any graph that contains $C_4$, $C_5$, or $P_5$ as an induced subgraph
does not satisfy the first property. Hence it remains to show that every graph that does not satisfy the first
property contains either a $C_4$, a $C_5$, or a $P_5$ as induced subgraph. 

Consider a connected graph $G=(V,E)$. Every minimum vertex cover of which induces at least two connected components. 
Pick such a minimum vertex cover $\VC\subset V$ that induces the smallest number of connected components. 
There must exist two subsets $A,B\subseteq \VC$ inducing two disjoint connected components, and a vertex $v$, such that $G[A\cup B\cup \{v\}]$ is connected, by Lemma~\ref{dist2}.

Consider the breadth-first search (BFS) trees $T_A\subseteq E$ in $G[\{v\}\cup A]$, and $T_B\subseteq E$ in $G[\{v\}\cup B]$, both rooted at $v$. If both trees
have height at least two, then there is an induced $P_5$. Hence at least one of the trees, say $T_B$, has height one, that is, $N(v)\cap B = B$. 
Now we consider the set $\VC' := (\VC\setminus \{w\})\cup \{v\}$, where $w$ is an arbitrary vertex of $B$. Since the number of connected components in $G[\VC']$ is strictly less than the number of connected components in $G[\VC]$, and $\VC'$ is not bigger than $\VC$, the new set $\VC'$ cannot be a vertex cover. Therefore, there must exist a vertex $x\notin \VC$, such that $wx\in E$ 
is not covered by $\VC'$. Note that $xv\notin E$ (otherwise it would be covered by $\VC'$). If $x$ is adjacent to a vertex $t\in A$ that is itself adjacent to $v$, then we have found a $C_4$. If $x$ is adjacent to a vertex $t\in A$ that is not adjacent to $v$, then, using the shortest path from $v$ to $t$ in $T_A$, we find a cycle of length at least 5. 

Hence there remains the case where $x$ is not adjacent to any vertex in $A$. In that case, provided the height of $T_A$ is at least two, we can find a $P_5$. If the height of $T_A$ is exactly one, then $N(v)\cap A=A$, and we can do the same reasoning as above, and show there is a vertex $y\notin \VC$ adjacent to a vertex $z\in A$. Similarly, we can assume that $y$ is not adjacent to any vertex in $B$. Now if $x=y$, we have found a $C_4$. Otherwise, the path going from $x$ to $y$ through $A,v$, and $B$ is an induced $P_5$. 
\end{proof}

\subsection{PoC-Near-Perfect graphs}

Let $\VC$ be a vertex cover of a graph $G$.
Let $\VC'$ be the vertex set of a connected component of $G[\VC]$.
We define $P_{\VC}(\VC')$ to be the set of vertices $v \in V(G)$ such that $N(v) \cap \VC \subseteq \VC'$.
It is clear that $\VC' \subseteq P_{\VC}(\VC')$.

To prove Theorem \ref{thm PoC 4/3}, we need to use the following lemma.

\begin{lemma}\label{simple-argumentation}
Let $S_1, S_2, \dots, S_k$ be the vertex sets of connected components of a vertex cover $\VC$.
There exists at least one $P_{\VC}(S_i)$which is not a cutset of $G$, i.e. $G[V(G) \setminus P_{\VC}(S_i)]$ is always connected.
\end{lemma}
\begin{proof}
We consider the new following graph $H$ defined by $$V(H) = \{ P_{\VC}(S_i) | i = 1, \dots, k \}$$ and $$E(H) = \{ P_{\VC}(S_i)P_{\VC}(S_j) | N(P_{\VC}(S_i)) \cap N(P_{\VC}(S_j)) \neq \emptyset  \}.$$
Note that the sets $P_{\VC}(S_i)$, $1 \leqslant i \leqslant k$, are disjoint and induce a connected subgraph of $G$ each. Because $\VC$ is a vertex cover, $H$ is connected. Because every connected graph contains a no cutvertex, there exists at least one $P_{\VC}(S_i)$ which is not a cutvertex of $H$. Therefore, $P_{\VC}(S_i)$ is not a cutset of $G$. 
\end{proof}

\begin{proof}[Theorem \ref{thm PoC 4/3}]
Since the \PoC{} of $P_5$ and $C_4$ equals $3/2$, any graph that contains $C_4$ or $P_5$ as an induced subgraph does not satisfy the first property. 
Hence, it remains to show that every graph that does not satisfy the first property contains either a $C_4$ or a $P_5$ as induced subgraph. 

Let $G$ be a $(P_5, C_4)$-free graph. 
The proof is by induction on the number of components of a minimum vertex cover, say $k$. 
Let $\VC$ be such a vertex cover of $G$.
Let $S_1,S_2, \dots, S_k$ be the vertex set of the connected components of $G[\VC]$.

If $\VC$ is connected ($k=1$), then $\tau_c/\tau = 1$. 

If $k = 2$, i.e. $S_1$ and $S_2$ are connected components of $G[\VC]$, we have a vertex $x$ adjacent to $S_1$ and $S_2$, by Lemma~\ref{dist2}. Let $s_1 \in S_1$ and $s_2 \in S_2$ be two vertices such that the distance between $s_i$ and $x$ in $S_i$ is maximum. In particular, $S_1 \cup \{x\} \cup (S_2 \setminus \{s_2\})$ and $S_2 \cup \{x\} \cup (S_1 \setminus \{s_1\})$ are connected. If $S_1 \cup \{x\} \cup (S_2 \setminus \{s_2\})$ or $S_2 \cup \{x\} \cup (S_1 \setminus \{s_1\})$ is a vertex cover, then $\tau_c(G)/\tau(G) = 1$. 
Otherwise there are two edges $x_1s_1$ and $x_2s_2$  with $x_1, x_2 \notin \VC \cup \{x\}$.
If $x_1 = x_2$, $G[ \VC \cup \{x, x_1\}]$ has an induced $C_4$, if the distance between $s_i$ and $x$, for $i=1,2$, is exactly 1, and an induced $P_5$, if the distance between $s_i$ and $x$, for $i=1,2$, is at least two. Without loss of generality, we can suppose that the distance between $s_1$ and $x$ is one and the distance between $s_2$ and $x$ is two, i.e. $\tau(G) \geqslant 3$. Therefore, $$\frac{\tau_c(G)}{\tau(G)} \leqslant \frac{\tau(G)+1}{\tau(G)} = 1 + \frac{1}{\tau(G)} \leqslant \frac{4}{3}.$$
Otherwise, $x_1$ and $x_2$ are different. Moreover $x_1$ cannot be adjacent to $x_2$ because $\VC$ is a vertex cover. 
Hence, $G[\VC \cup \{x, x_1, x_2\}]$ contains a $P_5$. 
Thus, we obtain a contradiction in every case.

If $S_1$, $S_2$ and $S_3$ are connected components of $G[\VC]$, we can suppose, without loss of generality, that there exists $x_1 \in N(S_1) \cap N(S_2)$ and $x_2 \in N(S_2) \cap N(S_3)$, by Lemma~\ref{dist2}.
The vertex $x_2$ is adjacent to $S_1$ (or $x_1$ is adjacent to $S_3$), otherwise $G[\VC \cup \{x_1,x_2\} ]$ contains a $P_5$.
Thus, $\VC \cup \{x_2\}$, resp. $\VC \cup \{x_1\}$, is a connected vertex cover. 
Hence, $$\frac{\tau_c(G)}{\tau(G)} \leqslant \frac{|\VC|+1}{|\VC|} \leqslant \frac{4}{3}.$$

Now $k \geqslant 4$ and we assume that $\tau_c \leqslant 4/3 \tau$ holds for every connected $(P_5,C_4)$-free graph with a minimum vertex cover of at most $k-3$ connected components. Let $S_1, S_2, S_3, \dots S_k$ be the vertex sets of the connected components of $G[\VC]$.
By Lemma~\ref{simple-argumentation} at least one of these sets, say $P_{\VC}(S_3)$, is not a cutset of $G$. 
By applying twice the Lemma~\ref{simple-argumentation}, two more of these sets, say $P_{\VC}(S_2)$, resp.~$P_{\VC}(S_1)$, are not a cutset of $G[V \setminus P_{\VC}(S_3)]$, resp.~$G[V \setminus (P_{\VC}(S_2) \cup P_{\VC}(S_3))]$. 
Let $\VC'= \VC \setminus (S_1 \cup S_2 \cup S_3)$ and note that $\VC'$ is a minimum vertex cover of $G'= G[V \setminus (P_{\VC}(S_1) \cup P_{\VC}(S_2) \cup P_{\VC}(S_3))]$. 
By the induction hypothesis, there is a minimum connected vertex cover of $G'$, say $\CVC'$, with $|\CVC'| \leqslant 4/3 |\VC'|$.

We show that there exists a connected vertex cover $\CVC$ of $G$ with $|\CVC| \leqslant |S_1|+|S_2|+|S_3|+|\CVC'|+1$, built from $S_1, S_2, S_3$ and $\CVC'$. 
Indeed, we have $$ \frac{\tau_c(G)}{\tau(G)} \leqslant \frac{|S_1|+|S_2|+|S_3|+1+|\CVC'|}{|S_1|+|S_2|+|S_3|+|\VC'|} \leqslant \max\left(\frac{|S_1|+|S_2|+|S_3|+1}{|S_1|+|S_2|+|S_3|},\frac{|\CVC'|}{|\VC'|}\right) \leqslant \frac{4}{3}.$$
We refer to $\CVC'$ as $S_4$ for ease of writing. 
We observe that the set $V(G) \setminus (S_1 \cup S_2 \cup S_3 \cup S_4)$ is an independent set because its complement is a vertex cover of $G$.
We complete the proof with the following case distinction.

\textbf{Case 1.} There exists one component, say $S_1$, such that the other connected components are a distance $2$ from $S_1$. 
Let $x_i$ be a vertex adjacent to $S_1$ and $S_i$, for $i= 2, 3, 4$.

\textbf{Case 1.1.} The $x_i$ are mutually distinct. 
Since $G[S_1 \cup S_3 \cup S_4 \cup \{x_3,x_4\} ]$ contains an induced $P_5$, we are in the next case. 

\textbf{Case 1.2.} Two of the $x_i$ are equal, and the third one is distinct from them. 
We can suppose without loss of generality that $x_3 = x_4$. 
The path $S_2x_2S_1x_3S_3$ forms again a $P_5$. 
If there is an edge between $x_3$ and $S_2$, we take $x_3$ to connect $S_1, S_2, S_3$ and $S_4$.
Otherwise, there must be an edge between $x_2$ and $S_3$. 
But then, we have an induced $P_5$ in $G[S_2 \cup S_3 \cup S_4 \cup \{x_2,x_3\}]$, a contradiction.

\textbf{Case 1.3.} It holds that $x_2 = x_3 = x_4$. 
We have immediately one vertex to connect $S_1, S_2, S_3$ and $S_4$.

\textbf{Case 2.} Up to a renaming of the $S_i$, the distance between $S_i$ and $S_{i+1}$ is $2$, $i= 1, 2, 3$. 
Let $x_i$ be a vertex adjacent to $S_i$ and $S_{i+1}$, $i =1,2,3$. 
Because $G$ is $P_5$-free, $S_1$ must be adjacent to $x_2$ or $x_1$ must be adjacent to $S_3$. 
Hence, we are in Case 1. 
\end{proof}

To prove Theorem \ref{thm PoC 3/2}, we need to use the following lemma.

\begin{lemma}\label{withoutC7}
Let $G$ be a $(C_6,P_7,\Delta_1,\Delta_2)$-free graph, let $\VC$ be a vertex cover of $G$ such that $G[\VC]$ has exactly three connected components, and let $G$ contain an induced cycle of length $7$ intersecting all connected components of $G[\VC]$. 
Then there exists a connected vertex cover $\CVC$ such that 
\begin{center}
\begin{tabular}{ll}
  $|\CVC| \leqslant |\VC|+1$  & \ \ if $|\VC| > 4,$ \\
  $|\CVC| = 6$  &\ \ if $|\VC| = 4.$ \\
\end{tabular}
\end{center}
\end{lemma}

\begin{proof}
Let $x_1x_2x_3x_4x_5x_6x_7$ be an induced cycle intersecting the three connected components of $\VC$, say $S_1, S_2$ and $S_3$. 
Without loss of generality, we can suppose $x_1 \in S_1$, $x_3 \in S_2$ and $x_5, x_6 \in S_3$ (see Fig~\ref{fig:three cc with C7}). We can assume that no vertex is adjacent to 
$S_1, S_2$ and $S_3$.

\begin{figure}[!ht]
  \centering
\begin{tikzpicture}
\tikzstyle{every node}=[draw,circle,fill=black,minimum size=4pt,inner sep=0pt]
\node(x1) at (2,1){};
\node(x2) at (0,0){};
\node(x3) at (-2,1){};
\node(x4) at (-2,2){};
\node(x5) at (-1,3){};
\node(x6) at (1,3){};
\node(x7) at (2,2){};
\node [black,right=1mm, draw=none,fill=none] at (x1) {$x_1$};
\node [black,below=1mm, draw=none,fill=none] at (x2) {$x_2$};
\node [black,left=1mm, draw=none,fill=none] at (x3) {$x_3$};
\node [black,left=1mm, draw=none,fill=none] at (x4) {$x_4$};
\node [black,above=0.4mm, draw=none,fill=none] at (x5) {$x_5$};
\node [black,above=0.4mm, draw=none,fill=none] at (x6) {$x_6$};
\node [black,right=1mm, draw=none,fill=none] at (x7) {$x_7$};
\draw (x1) -- (x2) -- (x3) -- (x4) -- (x5) -- (x6) -- (x7) -- (x1);
\draw (x1) circle (0.6cm);
\draw (x3) circle (0.6cm);
\node[draw = black,ellipse,fill=none,inner sep=7pt,fit= (x5) (x6)]{};
\node [black, draw=none,fill=none] at (2.9,1) {$S_1$};
\node [black, draw=none,fill=none] at (-2.9,1) {$S_2$};
\node [black, draw=none,fill=none] at (2.2,3) {$S_3$};
\end{tikzpicture}
  \caption{G contains an induced $C_7$ intersecting all components of a vertex cover.}
\label{fig:three cc with C7}
\end{figure}
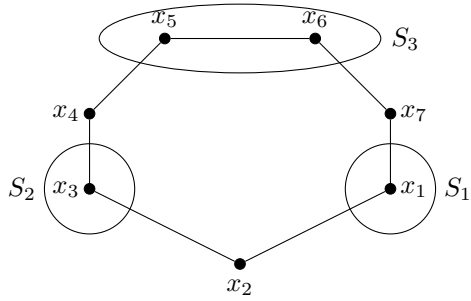

If $|N(x_1) \cap N(x_3) | \geqslant 2$, $|N(x_1) \cap N(x_6) | \geqslant 2$, or $|N(x_5) \cap N(x_3) | \geqslant 2$, then we have an induced subgraph $\Delta_2$.

Otherwise $N(x_1) \cap N(x_3)  = \{x_2\}$ and $N(x_1) \cap N(x_6)  = \{x_7\}$ and $N(x_5) \cap N(x_3)  = \{x_4\}$. 
We distinguish several cases, depending on the cardinality of $S_3$. 

The first case is $|S_3| = 2$. 
If $|S_1|=|S_2|=1$, i.e., $|\VC | = 4$, we have the connected vertex cover $S_1 \cup S_2 \cup S_3 \cup \{x_2, x_4\}$ of six vertices.

We can suppose $|S_2| >1$. 
Then every vertex of $S_2$ is adjacent to both $x_2$ and $x_4$. 
Indeed, if an edge $yz$ of $S_2$ has one endvertex, say $y$, adjacent to both $x_2$ and $x_4$, then $z$ must be adjacent to both $x_2$ and $x_4$, otherwise $G$ contains a $P_7$. 
Let $x$ a vertex of $S_2 \setminus \{x_3\}$ which is not a cutvertex of $G[S_2]$, i.e. $Y = S_1 \cup S_3 \cup (S_2 \setminus \{x\}) \cup \{x_2,x_4\}$ induces a connected graph. If $Y$ is not a vertex cover, there exists a vertex $t \notin Y$ adjacent to $x$. 
Note that $t$ is distinct from $x_7$, because no vertex is adjacent to $S_1,S_2$ and $S_3$, and $t$ is not adjacent to $x_1$ or $x_5$, because $G$ is $\Delta_2$-free. 
Moreover $t$ is not adjacent to $x_6$ since $G$ is $C_6$-free. 
Therefore we have an induced $P_7$ subgraph (see Fig~\ref{fig:3cc with P7}).

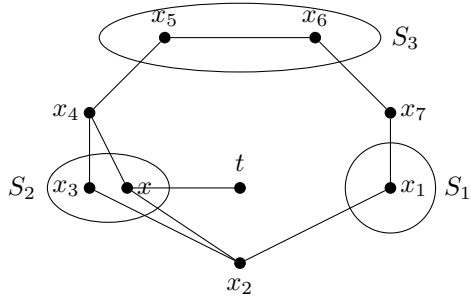
\begin{figure}[!ht]
  \centering
\begin{tikzpicture}
\tikzstyle{every node}=[draw,circle,fill=black,minimum size=4pt,inner sep=0pt]
\node(x1) at (2,1){};
\node(x2) at (0,0){};
\node(x3) at (-2,1){};
\node(x4) at (-2,2){};
\node(x5) at (-1,3){};
\node(x6) at (1,3){};
\node(x7) at (2,2){};
\node(x) at (-1.5,1){};
\node(t) at (0,1){};
\node [black,right=1mm, draw=none,fill=none] at (x1) {$x_1$};
\node [black,below=1mm, draw=none,fill=none] at (x2) {$x_2$};
\node [black,left=1mm, draw=none,fill=none] at (x3) {$x_3$};
\node [black,above=2mm, right=1mm, draw=none,fill=none] at (x) {$x$};
\node [black,left=1mm, draw=none,fill=none] at (x4) {$x_4$};
\node [black,above=0.4mm, draw=none,fill=none] at (x5) {$x_5$};
\node [black,above=0.4mm, draw=none,fill=none] at (x6) {$x_6$};
\node [black,right=1mm, draw=none,fill=none] at (x7) {$x_7$};
\node [black,above=2mm, draw=none,fill=none] at (t) {$t$};
\draw (x1) -- (x2) -- (x3) -- (x4) -- (x5) -- (x6) -- (x7) -- (x1);
\draw (x) -- (t);
\draw (x2)--(x) -- (x4);
\draw (x1) circle (0.6cm);
\node[draw = black,ellipse,fill=none,inner sep=7pt,fit= (x5) (x6)]{};
\node[draw = black,ellipse,fill=none,inner sep=7pt,fit= (x3) (x)]{};
\node [black, draw=none,fill=none] at (2.9,1) {$S_1$};
\node [black, draw=none,fill=none] at (-2.9,1) {$S_2$};
\node [black, draw=none,fill=none] at (2.2,3) {$S_3$};
\end{tikzpicture}
  \caption{G contains an induced $C_7$ and $P_7$ intersecting all components of a vertex cover.}
\label{fig:3cc with P7}
\end{figure}

In the second case, $|S_3| > 2$. If there exists a vertex in $S_3 \setminus \{x_5,x_6\}$ which is adjacent to neither $x_4$ nor $x_7$, we have an induced $P_7$ subgraph.
Let $y \in S_3\setminus \{x_5,x_6\}$ such that $y$ is not a cutvertex of $G[S_3]$. 
We can suppose that $y$ is adjacent to $x_4$. Because $G$ is $C_6$-free, $y$ is not adjacent to $x_7$. If $Y = S_1\cup S_2 \cup (S_3 \setminus \{y\}) \cup \{x_7, x_4\}$ is not a vertex cover, there exists a vertex $t \notin Y$ adjacent to $y$. 
Note that $t$ is distinct from $x_2$, because no vertex is adjacent to $S_1, S_2$ and $S_3$ (see Fig~\ref{fig:cc with C6andD2}).
If $t$ is adjacent to $x_1$ (resp.~$x_3$), we have an induced $C_6$ subgraph (resp.~$\Delta_2$).

\begin{figure}[!ht]
  \centering
\begin{tikzpicture}
\tikzstyle{every node}=[draw,circle,fill=black,minimum size=4pt,inner sep=0pt]
\node(x1) at (2,1){};
\node(x2) at (0,0){};
\node(x3) at (-2,1){};
\node(x4) at (-2,2){};
\node(x5) at (-1,3){};
\node(x6) at (1,3){};
\node(x7) at (2,2){};
\node(y) at (0,2.6){};
\node(t) at (0,1.6){};
\node [black,right=1mm, draw=none,fill=none] at (x1) {$x_1$};
\node [black,below=1mm, draw=none,fill=none] at (x2) {$x_2$};
\node [black,left=1mm, draw=none,fill=none] at (x3) {$x_3$};
\node [black,right=2mm,  draw=none,fill=none] at (y) {$y$};
\node [black,left=1mm, draw=none,fill=none] at (x4) {$x_4$};
\node [black,above=0.4mm, draw=none,fill=none] at (x5) {$x_5$};
\node [black,above=0.4mm, draw=none,fill=none] at (x6) {$x_6$};
\node [black,right=1mm, draw=none,fill=none] at (x7) {$x_7$};
\node [black,below=2mm, draw=none,fill=none] at (t) {$t$};
\draw (x1) -- (x2) -- (x3) -- (x4) -- (x5) -- (x6) -- (x7) -- (x1);
\draw (y) -- (t);
\draw (x4)--(y) ;
\draw (x1) circle (0.6cm);
\draw (x3) circle (0.6cm);
\node[draw = black,ellipse,fill=none,inner sep=7pt,fit= (x5) (x6) (y)]{};
\node [black, draw=none,fill=none] at (2.9,1) {$S_1$};
\node [black, draw=none,fill=none] at (-2.9,1) {$S_2$};
\node [black, draw=none,fill=none] at (2.2,3) {$S_3$};
\end{tikzpicture}
  \caption{$Y = S_1 \cup S_2 \cup (S_3 \setminus \{y\}) \cup \{x_4, x_7\}$ is not a vertex cover of $G$.}
\label{fig:cc with C6andD2}
\end{figure}
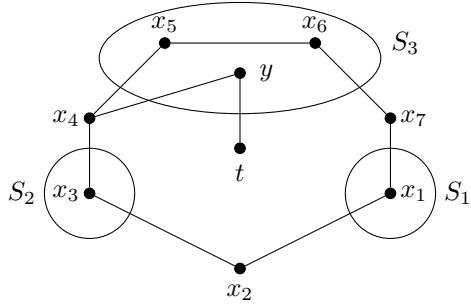
Otherwise we have an induced $P_7$ subgraph.
\end{proof}

\begin{proof}[Theorem \ref{thm PoC 3/2}]
If $G$ contains one of the four forbidden induced subgraphs, say $H$, then $\tau_c(H)/\tau(H)  = 5/3.$
It remains to prove that the Price of Connectivity of a $(P_7,C_6, \Delta_1, \Delta_2)$-free graph is bounded by $3/2$. 
So let $G$ be a $(P_7,C_6, \Delta_1, \Delta_2)$-free graph. 
The proof is by induction on the number of connected components of a minimum vertex cover. 
Let $\VC$ be a minimum vertex cover of $G$ with the minimum number of connected components, say $k$. 

If $\VC$ is connected ($k=1$), then $\tau_c/\tau = 1$. 

If $k = 2$, by Lemma~\ref{dist2}, $$\tau_c/\tau \leqslant \frac{|\VC|+1}{|\VC|} \leqslant 1 + \frac{1}{|\VC|} \leqslant 1 + \frac{1}{2} = \frac{3}{2}.$$

Now let $k \geqslant 3$.
We may assume that $\tau_c \leqslant 3/2 \tau$ holds for every $(P_7,C_6, \Delta_1, \Delta_2)$-free graph with a minimum vertex cover of at most $k-2$ connected components. 
Let $S_1, S_2, S_3, \dots, S_k$ be the vertex set of the connected components of $G[\VC]$. 
By Lemma~\ref{simple-argumentation}, we may assume that the set $P_{\VC}(S_2)$, is not a cutset of $G$, and that the set $P_{\VC}(S_1)$ is not a cutset of $G' = G[V \setminus P_{\VC}(S_2)]$.
Note that the set $\VC' = \VC \setminus (S_1 \cup S_2)$ is a minimum vertex cover of the graph $G'' = G[ V \setminus (P_{\VC}(S_1) \cup P_{\VC}(S_2))]$. 
By the induction hypothesis, there is a minimum connected vertex cover of $G''$, say $\CVC'$, with $|\CVC'| \leqslant 3/2 |\VC'|$.

We show that there exists a connected vertex cover $\CVC$ of $G$ such that $|\CVC| \leqslant |S_1|+|S_2|+|\CVC'|+1$, built from $S_1, S_2$ and $\CVC'$. Indeed, we have 
$$ \frac{\tau_c(G)}{\tau(G)} \leqslant \frac{|S_1|+|S_2|+1+|\CVC'|}{|S_1|+|S_2|+|\VC'|} \leqslant \max\left(\frac{|S_1|+|S_2|+1}{|S_1|+|S_2|},\frac{|\CVC'|}{|\VC'|}\right) \leqslant \frac{3}{2}.$$
We refer to the set $\CVC'$ as $S_3$ for the ease of writing. 
We can suppose that there does not exist any single vertex to connect $S_1, S_2,$ and $S_3$.
Without loss of generality, there is a vertex $x_i$ adjacent only to $S_i$ and $S_{i+1}$, $i=1,2$, such that $x_1 \neq x_2$ (see Fig.~\ref{fig:three cc to connect}). 

\begin{figure}
\centering
\begin{minipage}{.49\textwidth}
 \subfloat[Initial case]{\label{fig:three cc to connect}
\begin{tikzpicture}
\tikzstyle{every node}=[draw,circle,fill=black,minimum size=4pt,inner sep=0pt]
\node(x1) at (-1,1){};
\node(x2) at (1,1){};
\node [black,left=1mm, draw=none,fill=none] at (x1) {$x_1$};
\node [black,right=1mm, draw=none,fill=none] at (x2) {$x_2$};
\draw (-1.7,0.3) -- (x1) -- (-0.3,1.7);
\draw (1.7,0.3) -- (x2) -- (0.3,1.7);
\draw (-2,0) circle (0.6cm);
\draw (2,0) circle (0.6cm);
\draw (0,2) circle (0.6cm);
\node [black, draw=none,fill=none] at (-2,0) {$S_1$};
\node [black, draw=none,fill=none] at (0,2) {$S_2$};
\node [black, draw=none,fill=none] at (2,0) {$S_3$};
\end{tikzpicture}}
\end{minipage}
\begin{minipage}{.49\textwidth}
 \subfloat[Two private edges of a vertex cover]{\label{fig:three cc with private edges}
\begin{tikzpicture}
\tikzstyle{every node}=[draw,circle,fill=black,minimum size=4pt,inner sep=0pt]
\node(x1) at (-1,1){};
\node(x2) at (1,1){};
\node(y) at (-0.4,0){};
\node(z) at (0.4,0){};
\node(y1) at (-1.6,0){};
\node(z3) at (1.6,0){};
\node [black,left=1mm, draw=none,fill=none] at (x1) {$x_1$};
\node [black,right=1mm, draw=none,fill=none] at (x2) {$x_2$};
\node [black,above=1mm, draw=none,fill=none] at (y) {$y$};
\node [black,above=1mm, draw=none,fill=none] at (z) {$z$};
\node [black,below=2mm, left = 0.1mm, draw=none,fill=none] at (y1) {$y_1$};
\node [black,below=2mm, right = 0.2mm, draw=none,fill=none] at (z3) {$z_3$};
\draw (-1.7,0.3) -- (x1) -- (-0.3,1.7);
\draw (1.7,0.3) -- (x2) -- (0.3,1.7);
\draw (y1) -- (y);
\draw (z3) -- (z);
\draw (-2,0) circle (0.6cm);
\draw (2,0) circle (0.6cm);
\draw (0,2) circle (0.6cm);
\node [black, draw=none,fill=none] at (-2.2,0) {$S_1$};
\node [black, draw=none,fill=none] at (0,2) {$S_2$};
\node [black, draw=none,fill=none] at (2.2,0) {$S_3$};
\end{tikzpicture}}
\end{minipage}

\label{fig:thm 3/2 1}
\caption{Three components of a vertex cover to connect by adding only one vertex.}
\end{figure}
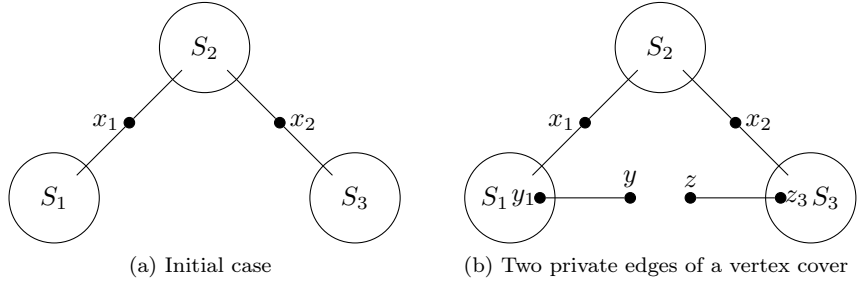

Note that $x_1$ and $x_2$ are not adjacent because $\VC$ is a vertex cover.
Let $y_1 \in S_1$ and $z_3 \in S_3$ be two vertices such that $y_1$ and $z_3$ are not cutvertices in $G[S_1\cup S_2 \cup S_3 \cup \{x_1,x_2\}]$.
If $S_1 \cup S_2 \cup (S_3 \setminus \{z_3\}) \cup \{x_1, x_2\}$ or $(S_1 \setminus \{y_1\}) \cup S_2 \cup S_3 \cup \{x_1, x_2\}$ is a vertex cover, then $\tau_c(G)/\tau(G) \leqslant 3/2$. 
Thus, there exist two edges, say $y_1y$ and $z_3z$, with $y, z \notin S_1 \cup S_2 \cup S_3 \cup \{x_1, x_2\}$ (see Fig.~\ref{fig:three cc with private edges}). 
Note that $y$ can be equal to $z$.

Now, we discuss on the adjacency of $y$ with $S_3$ and $S_2$.

\textbf{Case 1.} The vertex $y$ is adjacent to $S_3$.
Thus, $y$ is not adjacent to $S_2$.  
If the shortest induced cycle via $S_1 \cup S_2 \cup S_3 \cup \{x_1,x_2,y\}$ is of length $6$ or more than $8$, we have an induced $C_6$ or a $P_7$. 
Thus the shortest induced cycle via the three connected components has $7$ vertices. 
By Lemma \ref{withoutC7}, it is clear that $\tau_c(G) /\tau(G) \leqslant 3/2$ if $|S_1|+|S_2|+|S_3| > 4$. 
Otherwise $S_1, S_2$ and $S_3$ are three connected components of the initial vertex cover $\VC$. 
Thus, by Lemma \ref{withoutC7}, $\tau_c(G) / \tau(G) \leqslant 6/4 = 3/2$.

The cases that $z$ is adjacent to $S_1$ or $y=z$ are dealt with similarly.

\textbf{Case 2.} The vertex $y$ is adjacent to $S_2$ and $z$ is not adjacent to $S_2$.
Since $G[S_1 \cup S_2 \cup S_3 \cup \{x_1, x_2, y, z\}]$ does not contain $P_7$, there exists $t \in N(x_1) \cap N(x_2) \cap S_2$ and $t$ is adjacent to $y$. 
Hence, we have an induced $\Delta_2$.

\textbf{Case 3.} Both $y$ and $z$ are adjacent to $S_2$. 
Thus $y$ (resp. $z$) is not adjacent to $S_3$ (resp. $S_1$).
Let $P$ be a shortest path from $z$ to $y$ that goes through $S_3, \{x_2\}, S_2, \{x_1\}$, and $S_1$. 
If $P$ has $7$ vertices, then we have an induced $P_7$, $\Delta_1$ or $\Delta_2$ subgraph, depending on the adjacency of $y$ and $z$ with $S_2$.
If $P$ contains at least nine vertices, we have an induced $P_7$ subgraph in $G[S_1 \cup S_2 \cup S_3 \cup \{x_1, x_2\}]$.
Otherwise $P$ has exactly $8$ vertices. 
There are two possibilities.

\textbf{Case 3.1.} $S_1$ (or $S_3$) contains an edge of $P$ (see fig.~\ref{fig:edge in S1}). 
Thus we have an induced $P_7$ or $\Delta_2$ in $G[\{z,x_1,x_2\} \cup S_1 \cup S_2 \cup S_3\}]$, depending on the adjacency between $S_2$ and $z$.

\textbf{Case 3.2.} $S_2$ contains an edge of $P$ (see Fig.~\ref{fig:edge in S2}), say $vu$.
Then, if $z$ is not adjacent to $v$, $G$ contains a $P_7$ or a $\Delta_2$, depending on the adjacency between $z$ and $u$. 
Thus $z$ is adjacent to $v$. 
Hence, we have an induced $P_7$ or $\Delta_2$ subgraph in $G[S_1 \cup (S_2 \setminus \{u\}) \cup S_3 \cup \{x_1, x_2, y, z\}]$, depending on the adjacency between $y$ and $v$.

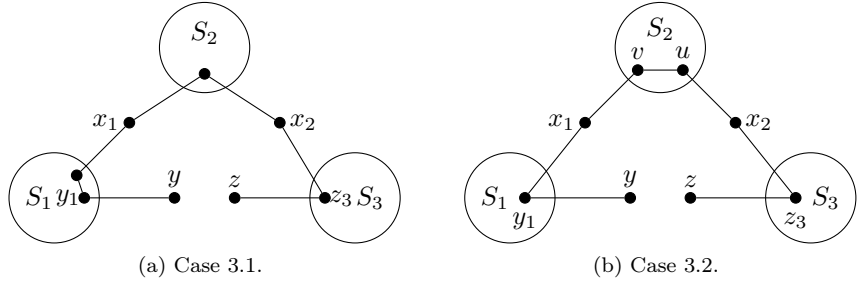
\begin{figure}
\centering
\begin{minipage}{.49\textwidth}
 \subfloat[Case 3.1.]{\label{fig:edge in S1}
\begin{tikzpicture}
\tikzstyle{every node}=[draw,circle,fill=black,minimum size=4pt,inner sep=0pt]
\node(x1) at (-1,1){};
\node(x2) at (1,1){};
\node(y) at (-0.4,0){};
\node(z) at (0.4,0){};
\node(y1) at (-1.6,0){};
\node(z3) at (1.6,0){};
\node(t) at (0,1.65){};
\node(r) at (-1.7,0.3){};
\node [black,left=1mm, draw=none,fill=none] at (x1) {$x_1$};
\node [black,right=1mm, draw=none,fill=none] at (x2) {$x_2$};
\node [black,above=1mm, draw=none,fill=none] at (y) {$y$};
\node [black,above=1mm, draw=none,fill=none] at (z) {$z$};
\node [black,below=2mm, left = 0.1mm, draw=none,fill=none] at (y1) {$y_1$};
\node [black,below=2mm, right = 0.2mm, draw=none,fill=none] at (z3) {$z_3$};
\draw (y1) -- (r) -- (x1) -- (t);
\draw (z3) -- (x2) -- (t);
\draw (y1) -- (y);
\draw (z3) -- (z);
\draw (-2,0) circle (0.6cm);
\draw (2,0) circle (0.6cm);
\draw (0,2) circle (0.6cm);
\node [black, draw=none,fill=none] at (-2.2,0) {$S_1$};
\node [black, draw=none,fill=none] at (0,2.2) {$S_2$};
\node [black, draw=none,fill=none] at (2.2,0) {$S_3$};
\end{tikzpicture}}
\end{minipage}
\begin{minipage}{.49\textwidth}
 \subfloat[Case 3.2.]{\label{fig:edge in S2}
\begin{tikzpicture}
\tikzstyle{every node}=[draw,circle,fill=black,minimum size=4pt,inner sep=0pt]
\node(x1) at (-1,1){};
\node(x2) at (1,1){};
\node(y) at (-0.4,0){};
\node(z) at (0.4,0){};
\node(y1) at (-1.8,0){};
\node(z3) at (1.8,0){};
\node(u) at (0.3,1.7){};
\node(v) at (-0.3,1.7){};
\node [black,left=1mm, draw=none,fill=none] at (x1) {$x_1$};
\node [black,right=1mm, draw=none,fill=none] at (x2) {$x_2$};
\node [black,above=1mm, draw=none,fill=none] at (y) {$y$};
\node [black,above=1mm, draw=none,fill=none] at (z) {$z$};
\node [black,below=1mm, draw=none,fill=none] at (y1) {$y_1$};
\node [black,below=1mm, draw=none,fill=none] at (z3) {$z_3$};
\node [black,above=1mm, draw=none,fill=none] at (u) {$u$};
\node [black,above=1mm, draw=none,fill=none] at (v) {$v$};
\draw (y1) -- (x1) -- (v) -- (u);
\draw (z3) -- (x2) -- (u);
\draw (y1) -- (y);
\draw (z3) -- (z);
\draw (-2,0) circle (0.6cm);
\draw (2,0) circle (0.6cm);
\draw (0,2) circle (0.6cm);
\node [black, draw=none,fill=none] at (-2.2,0) {$S_1$};
\node [black, draw=none,fill=none] at (0,2.25) {$S_2$};
\node [black, draw=none,fill=none] at (2.2,0) {$S_3$};
\end{tikzpicture}}
\end{minipage}

\label{fig:thm 3/2 2}
\caption{Three components of a vertex cover to connect by adding only one vertex.}
\end{figure}

\textbf{Case 4.} The vertex $y$ is adjacent to neither $S_2$ nor $S_3$.
We can suppose that $z$ is adjacent to neither $S_1$ nor $S_2$ (thus $y \neq z$).
Thus, $G$ contains a $P_7$.
\end{proof}

\subsection{PoC-Critical Graphs}

\subsubsection{PoC-Critical Chordal Graphs}

\begin{lemma}\label{no bridge}
Let $G$ be a critical graph. For every minimum vertex cover $\VC$ of $G$, there does not exist a bridge of $G$ with endvertices in $\VC$. 
\end{lemma}
\begin{proof}
Suppose there exists a bridge $xy$ with $x,y \in \VC$.
The removal of the edge $xy$ results in two connected subgraphs of $G$, which we denote by $G_1$ resp. $G_2$.
We can assume that $x \in V(G_1)$ and $y \in V(G_2)$.
Let $G_1'$ be the graph obtained from $G_1$ by attaching a pendant vertex to $x$.
Analogously let $G_2'$ be the graph obtained from $G_2$ by attaching a pendant vertex to $y$.

We observe that $\VC \cap V(G_1)$ is a vertex cover of $G_1'$ and $\VC \cap V(G_2)$ is a vertex cover of $G_2'$.
Thus
\begin{equation} \label{eqn tau ge tau + tau}
\tau(G) \geqslant \tau(G_1') + \tau(G_2').
\end{equation}

On the other hand, let $\CVCone$ be a connected vertex cover of $G_1'$ and $\CVCtwo$ be a connected vertex cover of $G_2'$.
We can assume that $\CVCone \subseteq V(G_1)$ and $\CVCtwo \subseteq V(G_2)$.
It is clear that $x \in \CVCone$ and $y \in \CVCtwo$.
Thus $\CVCone \cup \CVCtwo$ is a connected vertex cover of $G$.
Since $\CVCone \cap \CVCtwo = \emptyset$,
\begin{equation} \label{eqn tauc le tauc + tauc}
\tau_c(G) \leqslant \tau_c(G_1') + \tau_c(G_2').
\end{equation}

But (\ref{eqn tau ge tau + tau}) and (\ref{eqn tauc le tauc + tauc}) say that 
\begin{equation} \label{eqn PoC(G) le max PoC(G1),PoC(G2)}
\tau_c(G)/\tau(G) \leqslant \max \{\tau_c(G_1')/\tau(G_1') , \tau_c(G_2')/\tau(G_2')\}.
\end{equation}
Since both $G_1'$ and $G_2'$ are isomorphic to induced subgraphs of $G$, (\ref{eqn PoC(G) le max PoC(G1),PoC(G2)}) is a contradiction to the choice of $G$ to be critical.
\end{proof}

\begin{proof}[Theorem \ref{thm chordal critical}]
It is obvious that (ii) implies (iii). First, we prove that (iii) implies (i), that is, every critical chordal graph is a special tree. For this, let $G$ be a critical chordal graph. 

If the chordal graph $G$ is not a tree, then $G$ contains a triangle and every minimum vertex cover of $G$ contains at least two vertices of this triangle. Let $v$ be a vertex that is both in the triangle and in a minimum vertex cover. Then we have $\tau(G) = \tau(G-v)+1$ and also $\tau_c(G) \leqslant \tau_c(G-v)+1$, implying that $G$ is not critical. Therefore, $G$ is a tree.

Let $\VC$ be a minimum vertex cover of $G$.

First we show that $\VC$ is an independent set. Suppose there are $x, y \in \VC$ such that $xy \in E$. Since $G$ is a tree, $xy$ is a bridge, a contradiction with Lemma~\ref{no bridge}.

Now we show that every member of $V \setminus \VC$ has degree at most two.
For this, let $x \in V \setminus \VC$.
Suppose that $|N(x)| \geqslant 3$.
Let $X_1, X_2 , \ldots , X_k$ be the vertex sets of the connected components of $G-x$.
By assumption, $k \geqslant 3$.
Let 
$$H_1 = G - \bigcup_{i = 3}^k X_i$$ 
and
$$H_2 = G - (X_1 \cup X_2).$$

We observe that 
\begin{equation} \label{eqn tau ge tau + tau 2}
\tau(G) \geqslant \tau(H_1) + \tau(H_2).
\end{equation}

Since $x$ is a cutvertex of $H_1$, $x$ is contained in every connected vertex cover of $H_1$.
Therefore
\begin{equation} \label{eqn tau le tau + tau 2}
\tau_c(G) \leqslant \tau_c(H_1) + \tau_c(H_2).
\end{equation}

By the same argumentation from Lemma~\ref{no bridge}, (\ref{eqn tau ge tau + tau 2}) and (\ref{eqn tau le tau + tau 2}) yield a contradiction to the choice of $G$ to be critical.
This proves that every vertex of $V \setminus \VC$ has at most two neighbors.
By the discussion above, $\VC$ is an independent set and $G$ is a tree.
Moreover, the degree of every vertex in $\VC$ is at least two. Otherwise let $v$ be a vertex of $\VC$ with degree $1$ and let $u$ be the neighbor of $v$. Because $\VC$ is independent, $u \notin \VC$. Because $\VC$ is a vertex cover, every neighbor of $u$ is in $\VC$. Thus, $Y = (\VC \setminus \{v\}) \cup \{u\}$ is a minimum vertex cover but $Y$ is not independent, a contradiction.
We prove that $G$ is a special tree. In fact, the initial tree $H$ is defined as following :
$V(H) = \VC$ and $E(H) = \{uv| $there exists a path $P_{uv}$ in $V(G)\setminus \VC$ from $u$ to $v\}$. 
Because $\VC$ is a vertex cover of $G$, if $uv$ is an edge in $H$, then the length of the path $P_{uv}$ in $G$ is exactly $2$. Moreover, two $1$-degree vertices cannot have the same neighbor, because $G$ is critical. All in all, $G$ is a special tree.

Now, we show that (i) implies (ii), that is, every special tree is strongly critical.
Let $G$ be a special tree.
It is easy to see that $\tau_c(G)/\tau(G) =  2 - 1/\tau(G)$. If $G$ is not strongly critical, then there exists a proper subgraph $H$ of $G$ such that $\tau_c(H)/\tau(H) \geqslant \ \tau_c(G)/\tau(G).$ We can suppose that such an $H$ is minimal for inclusion. Thus $H$ is critical. By the previous argumentation, $H$ is a special tree. Therefore, $\tau_c(H) /\tau(H) = 2 - 1/\tau(H)$, but $2-1/\tau(G) > 2 - 1/\tau(H')$ for every proper special subtree $H'$ of $G$, a contradiction.
This completes the proof. 
\end{proof}

\subsubsection{PoC-Strongly-Critical Graphs}

Theorem \ref{thm strongly critical structure} follows from Lemma \ref{lem strongly critical => bipartite} and Lemma \ref{lem minimal with cutvertex => special} presented below. 

\begin{lemma} \label{lem strongly critical => bipartite}
Let $G$ be a strongly critical graph.
Then every minimum vertex cover of $G$ is an independent set.
In particular, $G$ is bipartite.
\end{lemma}
\begin{proof}
Let $G$ be a strongly critical graph and let $\VC$ be a minimum vertex cover of $G$.
Suppose that $\VC$ is not an independent set.
Thus there are two adjacent vertices in $\VC$, say $x$ and $y$.

By Lemma~\ref{no bridge}, $xy$ cannot be a bridge of $G$.
So $G-xy$ is connected.
Let $\CVC$ be a minimum connected vertex cover of $G - xy$.
Suppose that $\{x,y\} \cap \CVC \neq \emptyset$.
Then $\tau_c(G-xy) = \tau_c(G)$, in contradiction to the choice of $G$ to be strongly critical.
Thus $\{x,y\} \cap \CVC = \emptyset$.
Hence $\CVC \cup \{x\}$ is a minimum connected vertex cover of $G$ and, moreover, $y \notin \CVC \cup \{x\}$.

Let $A = N_G(y) \cap \VC$ and $B = N_G(y) \setminus \VC$.
As $x \in A$, $A \neq \emptyset$.
Since $\VC$ is a minimum vertex cover, $B \neq \emptyset$.
Let $G'$ be the graph obtained from $G$ by the removal of all edges joining $y$ to $B$.
Since $\CVC \cup \{x\}$ is a connected vertex cover of $G$ and $y \notin \CVC$, then $A \cup B \subseteq \CVC \cup \{x\}$ and $G'$ is connected.
As $\VC \setminus \{y\}$ is a vertex cover of $G'$, $\tau(G') < \tau(G)$.
Thus, by the choice of $G$, $\tau_c(G') \leqslant \tau_c(G) - 2$.
Let $\CVC'$ be a minimum vertex cover of $G'$.
Then $A \cap \CVC' \neq \emptyset$.
Therefore $\CVC' \cup \{y\}$ is a connected vertex cover of $G$, in contradiction to the fact that $|\CVC' \cup \{y\}| \leqslant \tau_c(G')+1 \leqslant \tau_c(G) - 1$.
This completes the proof. 
\end{proof}

\begin{lemma} \label{lem minimal with cutvertex => special}
Let $G$ be a strongly critical graph.
If $G$ has a cutvertex, it is a special tree.
\end{lemma}
\begin{proof}
Let $G=(V,E)$ be a strongly critical graph with a cutvertex.
Suppose that $G$ is not a tree.
Thus $G$ has a non-trivial block.
We can pick a cutvertex $x$ and an edge $e$ incident to $x$ in this block.
The graph $G-e$ is connected, by the choice of $e$.
Every connected vertex cover of $G-e$ contains $x$, as $x$ is a cutvertex of $G-e$.
Hence, every connected vertex cover of $G-e$ covers $e$.
Thus $\tau_c(G-e) \geqslant \tau_c(G)$, in contradiction to the choice of $G$ to be strongly critical.
Hence, $G$ is a tree. In particular, $G$ is chordal.

The conclusion then follows from Theorem~\ref{thm chordal critical}. 
\end{proof}


\end{document}